\documentclass{article}
\usepackage{graphicx} % Required for inserting images
\usepackage[letterpaper, left=1in, right=1in, top=0.9in, bottom=0.9in]{geometry}
\usepackage[american]{babel}
\usepackage[normalem]{ulem}
\usepackage{amsmath, amssymb, cases, amsthm}
\usepackage{thmtools}
\usepackage{tikz}
\usepackage[shortlabels]{enumitem}
\usepackage{mdframed}
\usepackage{bbm}
\usepackage{bm}
\usepackage{microtype}
\usepackage{xcolor}
\usepackage{makecell}
\usepackage{mathtools}
\usepackage{algorithmic}
\usepackage[procnumbered,ruled,vlined,linesnumbered]{algorithm2e}
\usepackage{float}
\usepackage{varwidth}
\usepackage{svg}
\usepackage{xfrac}
\usepackage{tikz-network}
\usetikzlibrary{decorations.pathreplacing}
\usepackage{subcaption} % For subfigures and labeling (a), (b)

\SetKwInput{KwData}{Input}
\SetKwInput{KwResult}{Output}
\SetKwInput{KwGlobalVar}{Global variables}

\usepackage[bookmarks,colorlinks,breaklinks]{hyperref}
\hypersetup{urlcolor=blue, colorlinks=true, citecolor=green!50!black, linkcolor=blue}
\usepackage[capitalize,nosort,nameinlink]{cleveref}
\declaretheorem[numberwithin=section,refname={Theorem,Theorems},Refname={Theorem,Theorems}]{theorem}

\declaretheorem[numberlike=theorem]{definition}

\declaretheorem[numberlike=theorem,style=remark]{remark}

\theoremstyle{definition}
\title{MMS Approximations Under Additive Leveled Valuations}

\author{Mahyar Afshinmehr\footnote{Sharif University of Technology, {mahyarafshinmehr@gmail.com}} \and  Mehrafarin Kazemi\footnote{Sharif University of Technology, {mehr.kzm@gmail.com}} \and Kurt Mehlhorn\footnote{Max Planck Institute for Informatics,  and Fachbereich Informatik, Universit\``at des Saarlandes, Saarland Informatics Campus, {mehlhorn@mpi-inf.mpg.de}}}
\date{August 2024}

\begin{document}

\maketitle

\begin{abstract}
    We study the problem of fairly allocating indivisible goods to a set of agents with additive leveled valuations. A valuation function is called leveled if and only if bundles of larger size have larger value than bundles of smaller size. The economics literature has well studied such valuations.
    We use the maximin-share (MMS) and EFX as standard notions of fairness. We show that an algorithm introduced by Christodoulou et al. (\cite{christodoulou2024fairtruthfulallocationsleveled}) constructs an allocation that is EFX and $\sfrac{\lfloor \frac{m}{n} \rfloor}{\lfloor \frac{m}{n} \rfloor + 1}\text{-MMS}$. In the paper, it was claimed that the allocation is EFX and $\sfrac{2}{3}\text{-MMS}$. However, the proof of the MMS-bound is incorrect. We give a counter-example to their proof and then prove a stronger approximation of MMS. 
\end{abstract}

\section{Introduction}

Allocating indivisible goods to a set of agents is a central problem in economics and computer science. Various notions of fairness, such as proportionality, envy-freeness, and maximin-share fairness, have been studied in the case of indivisible items. Although proportionality and envy-freeness provide existential guarantees for divisible goods, they cannot be obtained for indivisible items. Therefore, several relaxations of these notions, such as \emph{envy-freeness-up-to-any-good (EFX)} and (approximate) \emph{maximin-share (MMS)}, are in the interest of the fair division community.

\emph{Envy-freeness-up-to-any-good (EFX)} is a compelling and well-studied notion of fairness proposed by \cite{10.1145/3355902}. While it has attracted many interests, it is still open whether such allocations exist even for four agents with additive valuations. Therefore, EFX allocations have been studied in various restricted settings. \emph{Maximin-share (MMS)} is another well-studied fairness notion proposed by \cite{RePEc:ucp:jpolec:doi:10.1086/664613}, which is a natural extension of the Cut and choose protocol. An allocation is MMS fair if every agent receives a bundle with a value greater or equal to her maximin-share. Intuitively, the maximin-share corresponds to the maximum value an agent can guarantee after proposing an initial allocation and keeping the least desirable bundle for herself.

In this note, we focus on a specific class of valuation functions called additive leveled valuations. A valuation is additive if the value of a bundle of items is the sum of the values of its items. A valuation is leveled if a bundle with a larger cardinality is preferred over a bundle with a smaller cardinality. These valuations capture a simple economic behavior that prioritizes quantity over quality.

\subsection{Our Contribution}

We study the problem of fairly allocation a set of $m$ indivisible items to a set of $n$ agents. We use the notions of EFX and (approximate) MMS for fairness. We focus on the case that each agent has an additive leveled valuation function over sets of items. We show that an algorithm given in~\cite{christodoulou2024fairtruthfulallocationsleveled} guarantees both EFX and $\sfrac{\lfloor \frac{m}{n} \rfloor}{\lfloor \frac{m}{n} \rfloor + 1}\text{-MMS}$ for such instances. In the paper, $\sfrac{2}{3}$\text{-MMS} was claimed. We give a counter-example for the proof provided in \cite{christodoulou2024fairtruthfulallocationsleveled} and give a proof for a stronger approximation guarantee for such instances. Note that we assume $m \ge 2n$ since instances with $m < 2n$ admit a full MMS allocation as shown in Proposition 3.1 of \cite{christodoulou2024fairtruthfulallocationsleveled}.

\subsection{Related Work}

This section discusses prior works on EFX, MMS, and leveled valuations. The literature on the fair allocation of indivisible items is growing rapidly. Therefore, we only cover a small number of prior works in this area and refer the reader to a comprehensive survey such as \cite{Amanatidis_2023}.

\textbf{EFX.}
The idea of EFX as a fairness notion raises interesting challenges in fair allocation theory, leading to many open questions. \cite{DBLP:journals/corr/PlautR17} showed that EFX allocations exist when agents have the same valuations or have identical ordering over goods. Although EFX allocations exist for two agents with general valuations\footnote{Cut and choose protocol can be used to compute EFX allocations for two agents.}, this result cannot be generalized to three agents with general valuation.
\cite{DBLP:journals/corr/abs-2002-05119} showed that complete EFX allocations exist for three agents with additive valuation functions. This was later improved by \cite{DBLP:journals/corr/abs-2102-10654} and \cite{57690578beb44ed8a8923cc540ed32c6}.  
Other studies have examined cases with a limited number of items or agents and specific types of valuations (\cite{DBLP:journals/corr/abs-1909-07650}, \cite{DBLP:journals/corr/abs-2008-08798}, \cite{10.1287/moor.2022.0044}). A recent important finding is that EFX allocations may not exist for general monotone valuations in chore allocation situations (\cite{christoforidis2024pursuitefxchoresnonexistence}).

\textbf{MMS.}
\cite{10.1145/3140756} showed that MMS allocations may not exist under additive valuations. Later studies achieved strong approximation guarantees, with \cite{akrami2023breaking34barrierapproximate} recently surpassing the previous threshold of 3/4. The current impossibility results for additive valuations stand at $1 - \sfrac{1}{n^4}$ by \cite{feige2021tightnegativeexamplemms}. Furthermore, MMS allocations have been widely studied beyond additive valuations. \cite{DBLP:journals/corr/BarmanM17}, \cite{uziahu2023fairallocationindivisiblegoods}, \cite{10.1016/j.artint.2021.103633}, and \cite{10.1016/j.artint.2023.104049} have well-studied the class of complement-free valuations.

\textbf{Leveled Valuations.}
A typical example of leveled valuations is allocating offices to different university departments. \cite{DBLP:journals/corr/BabaioffNT17} studied the leveled preferences in the context of competitive equilibrium in markets, while \cite{DBLP:journals/corr/abs-2109-08671} studied the existence of EFX allocations for agents equipped with leveled valuations for chores.

\section{Models and Preliminaries}

In this section, we introduce our setting and notations.

\textbf{Model.} A set of indivisible goods, denoted by $M$ with $|M| = m$, will be allocated to a set of agents, denoted by $N = \{1, 2, \ldots, n\}$. Each agent like $i$ has a valuation function denoted by $V_i : 2^M \rightarrow R^{+}$.

\textbf{Valuation functions.} In this note, we focus on a specific class of valuations called \emph{additive leveled} valuations. 

\begin{definition}
    A valuation function $V$ is additive if, for any bundle of items $S$ we have:

    \begin{align*}
        V(S) = \sum\limits_{g \in S} V({g})
    \end{align*}
    
\end{definition}

\begin{definition}
    A valuation function $V$ is leveled if, for any two bundles $S$ and $T$ with $|S| > |T|$, it holds that $V(S) > V(T)$.
\end{definition}

\begin{definition}
    A valuation function is additive leveled if it is both additive and leveled.
\end{definition}

\textbf{Allocation.} An allocation is an $n\text{-partition}$ of items into $n$ bundles. We call $B = (B_1, B_2, \ldots, B_n)$ an allocation if $\bigcup\limits_{i \in N} B_i = M$ and for every $i, j \in N$ such that $i \neq j$, we have $B_i \cap B_j = \emptyset$. In this note, we use $B$ and $B_i$ to denote a full allocation and agent $i$'s bundle, respectively. $\Pi_n(M)$ denotes the set of all possible $n\text{-partitions}$.

\textbf{Fairness notions.} Fairness notions can be divided into two categories called as \emph{envy-based} or \emph{share-based}. 

\begin{definition}
    An allocation is called envy-Free (EF) if, for every pair of distinct agents like $i$ and $j$, we have $V_i(B_i) \ge V_i(B_j)$.
\end{definition}

It has been known that EF allocations do not necessarily exist. A simple example without an EF allocation is to consider a single good with a positive value and two identical agents. Therefore, the fair division community widely studies several relaxations of envy freeness. The strongest and most popular \emph{Envy} based notion is \emph{envy-freeness-up-to-any-good (EFX)}.

\begin{definition}
    An allocation is called envy-Free up to any good if, for every pair of distinct agents like $i, j$ and every item $g \in B_j$, we have $V_i(B_i) \ge V_i(B_j \setminus g)$.\footnote{\text{We sometimes use $g$ instead of $\{g\}$ for the ease of notation.}}
\end{definition}

From the share-based notions, we introduce the \emph{maximin-share (MMS)} notion. $\mu_i$ denotes the maximin-share of an agent $i$, defined as follows:

\begin{definition}
    For an agent $i$, the maximin-share is the maximum value she can obtain after proposing an $n\text{-partition}$ of goods and securing the worst bundle for herself. i.e., 

    \begin{align*}
        \mu_i = \max\limits_{B \in \Pi_n(M)} \min\limits_{j \in N} V_i(B_j)
    \end{align*}
\end{definition}

\begin{definition}
    An allocation $B$ is said to be maximin-share fair if $V_i(B_i) \ge \mu_i$ for all $i \in N$.
\end{definition}

Since it has been known that exact MMS allocations do not necessarily exist, we focus on multiplicative approximations of MMS.

\begin{definition}
    Let $\alpha \in (0, 1]$. An allocation $B$ is said to be $\alpha\text{-MMS}$ if $V_i(B_i) \ge \alpha\mu_i$ for all $i \in N$.
\end{definition}

\section{EFX and Approximate MMS Guarantees}

In this section, we show that the algorithm in~\cite{christodoulou2024fairtruthfulallocationsleveled}, constructs an allocation that is EFX and guarantees a $\sfrac{\lfloor\frac{m}{n}\rfloor}{\lfloor\frac{m}{n}\rfloor + 1}\text{-MMS}$

For completeness, we quote the algorithm and the proof that the generated allocation by the algorithm is EFX from~\cite{christodoulou2024fairtruthfulallocationsleveled}\footnote{Check Theoerem 4 and Algorithm 1 in~\cite{christodoulou2024fairtruthfulallocationsleveled}.}. 

\textbf{Description of the algorithm.}
Let $m = kn + r$ be the total number of items with $r$ being a nonnegative integer smaller than $n$ and let $k = \lfloor m/n \rfloor$. The algorithm runs sequentially. We begin by fixing a quota system in which the first $n - r$ agents in the sequence are set to pick their most favourite bundle of $\lfloor m/n \rfloor$ goods. Subsequently, the remaining agents choose their favourite subsets of items, each containing $\lfloor m/n \rfloor + 1$ goods. Therefore, the agents get to choose a fixed number of items from a feasible set according to a predefined order.

\begin{algorithm} 
\caption{EFX under Leveled Valuations}
\begin{algorithmic}[1] \label{alg}
    \REQUIRE An instance with leveled valuations, $m = kn + r$
    \ENSURE An EFX allocation
    \STATE Select an arbitrary picking order $\sigma = [\sigma_1, \sigma_2, \dots, \sigma_n]$
    \STATE Let the first $n - r$ agents (according to $\sigma$) pick their favourite bundle of $\lfloor m/n \rfloor$ available items
    \STATE Let the remaining $r$ agents pick their favourite bundle, each consisting of $\lfloor m/n \rfloor + 1$ items
    \RETURN Allocation
\end{algorithmic}
\end{algorithm}

\begin{theorem}[\cite{christodoulou2024fairtruthfulallocationsleveled}]
    Algorithm~\ref{alg} constructs an EFX allocation for leveled valuations. 
\end{theorem} 

\begin{proof}
    We divide the agents into two levels, namely level $L$ and level $H$; agents in the lower level possess an item less than those in the higher level. Clearly, agents in $H$ attain more value than agents in $L$; they are in fact envy-free towards them. Therefore, the EFX criterion can be  violated only among agents residing at different levels, i.e, an agent in $L$ may strongly envy an agent in $H$. Suppose, for the sake of contradiction, that EFX is violated, that is, there is an agent $i$ in $L$ that EFX-envies an agent $j$ in level $H$, i.e., there is $g \in B_j$ such that $i$ prefers $B_j \setminus g$ over $B_i$. But then $i$ would have picked $B_j \setminus g$ since she preceded $j$ in the picking order. 
\end{proof}

We come to the approximation bound for the value of the maximin-share. \cite{christodoulou2024fairtruthfulallocationsleveled}, in Proposition 5.1, claimed that the allocation is $\sfrac{2}{3}\text{-MMS}$. First, we exhibit a flaw in their proof. Then, we provide another proof for a stronger guarantee of approximate MMS. Note that we assume $m \ge 2n$ since instances with $m < 2n$ admit a full MMS allocation as shown in Proposition 3.1 of \cite{christodoulou2024fairtruthfulallocationsleveled}.

\subsection{Problem in Proof of Proposition 5.1 in \cite{christodoulou2024fairtruthfulallocationsleveled}}

This section examines the proof for Proposition 5.1 provided in \cite{christodoulou2024fairtruthfulallocationsleveled}. We quote the original proof given in \cite{christodoulou2024fairtruthfulallocationsleveled} and highlight an oversight in red, discussing why it is invalid. After that, we provide a counter-example for which the proof does not work. 

\vspace{5mm}

\textbf{Original Proof.}
``Consider an arbitrary allocation \( B = (B_1, \ldots, B_n) \) obtained via Algorithm \ref{alg}. Let \( i \) be an agent that receives \( B_i \), \( |B_i| = k = \lfloor m/n \rfloor \) and let allocation \( B' = (B'_1, \ldots, B'_n) \) denote her MMS allocation. Assume without loss of generality that \( k \geq 2 \); the case where \( m < 2n \) is handled by Proposition 3.1 via an SDQ. Note that since the valuations are additive, there must exist an item \( g \in B_i \) such that \( V_i(g) \leq \frac{1}{2}V_i(B_i) \) since \( k \geq 2 \). Otherwise, the items would add up to something larger than \( V_i(B_i) \). If the bundles of all agents are of size \( k \), then \( g \) belongs to a bundle, say \( B'_j \), in her MMS allocation \( B' \). We know that \( |B'_j \backslash g| = k - 1 \), and thus, due to leveled valuations, \( V_i(B'_j \backslash g) \leq V_i(B_i) \). Therefore, we have \( \mu_i \leq V_i(B_i) + V_i(g) \leq \frac{3}{2}V_i(B_i) \). Differently, some agents may receive one less item than others as shown in Algorithm \ref{alg}. But then, \( \mu_i \leq V_i(B'_j) \) for some bundle \( B'_j \) with \( |B'_j| = k + 1 \) that contains \( g \), \textcolor{red} {and thus, \( V_i(B'_j \backslash g ) \leq V_i(B_i) \) (since \( i \) precedes \( j \) in the order)}. We conclude that \( \mu_i \leq V_i(B'_j) \leq V_i(B_i) + V_i(g) \leq \frac{3}{2}V_i(B_i) \)."

\vspace{5mm}
% \subsubsection{Problem in your Proof}

\textbf{Brief Explanation of The Flaw in The Justification:}
The assumption in the proof, ``$V_i(B'_j \setminus g) \leq V_i(B_i)$", is not necessarily correct in the highlighted part. This statement, justified by ``since $i$ precedes $j$ in the order", incorrectly assumes comparability between the values of bundles with the same size across different allocations $B$ and $B'$. For the assumption to hold, the bundles being compared must be in $B$.

\vspace{5mm}

\textbf{Counter Example:}
Consider a scenario with three agents, $a_1$, $a_2$, and $a_3$, and ten items, $g_1, \ldots, g_{10}$, where items $g_1$, $g_2$, $g_3$ each have a value of $b$ and items $g_4, \ldots, g_{10}$ have a value of $a$ for every agent, with $b > a$. For an additive valuation with these values for singletons, it is sufficient to have $4a > 3b$ to be leveled as well.

Suppose that $B$ is the allocation given by running Algorithm \ref{alg}. The bundles for agents would then be:
\[
B_1 = \{g_1, g_2, g_3\}, \quad B_2 = \{g_4, g_5, g_6\}, \quad B_3 = \{g_7, g_8, g_9, g_{10}\}
\]
Suppose these bundles are assigned in the same order to agents $a_1$, $a_2$, and $a_3$. Consider another allocation $B'$ that is a MMS-partition for $a_2$:
\[
B'_1 = \{g_1, g_4, g_5\}, \quad B'_2 = \{g_2, g_7, g_8\}, \quad B'_3 = \{g_3, g_6, g_9, g_{10}\}
\]
We show that the highlighted part of the proposition is invalid for $a_2$. Assume $g = g_6$, where $V_2(g) \leq \frac{1}{2}V_2(\{g_4,g_5, g_6\})$ as desired. Since $g$ is in $B'_3$, the proof claims we must have:

\begin{align*} 
 V_2(B'_3 \setminus g) = V_2(\{g_3, g_9, g_{10}\}) \leq V_2(\{g_4, g_5, g_6\}) = V_2(B_2),
\end{align*}
a contradiction.

\subsection{Existence of $\sfrac{\lfloor\frac{m}{n}\rfloor}{\lfloor\frac{m}{n}\rfloor + 1}\text{-MMS}$ }
    
This section will prove a stronger approximation of MMS under additive leveled valuations. The theorem is as follows: 

 \begin{theorem}
     Consider an instance with $n$ agents and $m$ items. Let $m = kn + r$ where $k = \lfloor \sfrac{m}{n} \rfloor$ and $0 \le r < m$. Then, there exists a $\sfrac{k}{k + 1}\text{-MMS}$ allocation under additive leveled valuations.
 \end{theorem}

\begin{proof}
    The desired allocation is obtained via Algorithm \ref{alg}. Let $B = (B_1, \ldots, B_n)$ be the output of the Algorithm \ref{alg}, and let $B' = (B'_1, \ldots, B'_n)$ to be the MMS allocation for agent $i$. Also, denote the MMS value of agent $i$ by $\mu_i = \min\limits_{j \in [n]} V_i(B'_j)$. We only care about $k \ge 2$ since $m < 2n$  has an exact MMS allocation, as discussed in \cite{christodoulou2024fairtruthfulallocationsleveled}. Note that there are $n - r$ bundles of size $k$ and $r$ bundles of size $k + 1$ in both allocations $B$ and $B'$. Since our valuations are leveled, the bundle corresponding to the MMS value (the one with the minimum value for agent $i$ in allocation $B'$) has size $k$. If $|B_i| = k + 1$, then by leveled valuations we have that $V_i(B_i) \ge \mu_i$. Therefore, we only consider the case where $|B_i| = k$. We divide the problem into two cases as follows:

    \textbf{Case 1: $r = 0$.} Since the valuations are additive, there exists an item $g \in B_i$ where $V_i(g) \le \frac{1}{k} V_i (B_i)$. Assume $g \in B'_j$ in allocation $B'$. Note that we have $|B'_j \setminus g| = k - 1$. Therefore, we have:
    \begin{align*}
        V_i(B'_j \setminus g) \le V_i(B_i) \text{ and hence } V_i(B'_j) \le V_i(B_i) + V_i(g) \le \frac{k + 1}{k} V_i(B_i)
    \end{align*}
     Since $V_i(B'_j) \ge \mu_i$, allocation $B$ is $\frac{k}{k + 1}\text{-MMS}$ for agent $i$.

    \textbf{Case 2: $r > 0$.} We sort the items according to the valuation of agent $i$. Then, we start from the highest valued item and group every $k$ consecutive item (in the order) together. We denote the $j$-th group of items $A_j$. W.l.o.g, we can assume that agents $1$ to $n - r$ have bundles of size $k$ and have chosen their bundle in the lexicographic order.
    Now, let $g'$ be the least valued item of the group $A_i$. Since agent $i$ was the $i$-th agent who chose her bundle of size $k$, at most $k(i - 1)$ items from $\bigcup\limits_{j = 1}^{j = i} A_j$ are already chosen when $i$ comes to choose, which means there are still $k$ items left from this set. Therefore, we have  $V_i(B_i) \ge k \cdot V_i(g')$. Now, observe that $\mu_i \le V_i(A_1)$. Since $r > 0$, we have at least $k + 1$ items after $A_i$ in the mentioned ordering. Let $A'$ be the group of the first $k + 1$ items after $A_i$. Then $V_i(A_1) \le V_i(A') \le (k + 1) V_i(g') \le \frac{k+1}{k} V_i(B_i) $. Thus $B$ is $\frac{k}{k+1}$-MMS. 
\end{proof}

\begin{remark}
    Theorem 5 in \cite{christodoulou2024fairtruthfulallocationsleveled} establishes that Algorithm \ref{alg} satisfies EFX, Pareto Optimality, truthfulness, non-bossiness, and neutrality. In this note, we demonstrate that the algorithm can also guarantee a $\frac{\lfloor\frac{m}{n}\rfloor}{\lfloor\frac{m}{n}\rfloor + 1}\text{-MMS}$, enhancing these valuable properties.
\end{remark}

\section{Conclusion and Discussion}

We studied fairness under additive leveled valuations using two different notions of fairness. We suspect that studying leveled valuations using our analysis might be helpful and provide meaningful insights for other important classes of valuations in the economics literature.

\bibliographystyle{abbrv}
\bibliography{main}

\begin{thebibliography}{10}

\bibitem{57690578beb44ed8a8923cc540ed32c6}
H.~Akrami, N.~Alon, B.~Chaudhury, J.~Garg, K.~Mehlhorn, and R.~Mehta.
\newblock {EFX}: A simpler approach and an (almost) optimal guarantee via
  rainbow cycle number.
\newblock In {\em EC 2023 - Proceedings of the 24th ACM Conference on Economics
  and Computation}, EC 2023 - Proceedings of the 24th ACM Conference on
  Economics and Computation, page~61, United States, July 2023. Association for
  Computing Machinery.
\newblock Publisher Copyright: {\textcopyright} 2023 Owner/Author(s).; 24th ACM
  Conference on Economics and Computation, EC 2023 ; Conference date:
  09-07-2023 Through 12-07-2023.

\bibitem{akrami2023breaking34barrierapproximate}
H.~Akrami and J.~Garg.
\newblock Breaking the $3/4$ barrier for approximate maximin share, 2023.

\bibitem{Amanatidis_2023}
G.~Amanatidis, H.~Aziz, G.~Birmpas, A.~Filos-Ratsikas, B.~Li, H.~Moulin, A.~A.
  Voudouris, and X.~Wu.
\newblock Fair division of indivisible goods: Recent progress and open
  questions.
\newblock {\em Artificial Intelligence}, 322:103965, Sept. 2023.

\bibitem{DBLP:journals/corr/abs-1909-07650}
G.~Amanatidis, A.~Ntokos, and E.~Markakis.
\newblock Multiple birds with one stone: Beating 1/2 for {EFX} and {GMMS} via
  envy cycle elimination.
\newblock {\em CoRR}, abs/1909.07650, 2019.

\bibitem{DBLP:journals/corr/BabaioffNT17}
M.~Babaioff, N.~Nisan, and I.~Talgam{-}Cohen.
\newblock Competitive equilibria with indivisible goods and generic budgets.
\newblock {\em CoRR}, abs/1703.08150, 2017.

\bibitem{DBLP:journals/corr/BarmanM17}
S.~Barman and S.~K.~K. Murthy.
\newblock Approximation algorithms for maximin fair division.
\newblock {\em CoRR}, abs/1703.01851, 2017.

\bibitem{DBLP:journals/corr/abs-2102-10654}
B.~Berger, A.~Cohen, M.~Feldman, and A.~Fiat.
\newblock (almost full) {EFX} exists for four agents (and beyond).
\newblock {\em CoRR}, abs/2102.10654, 2021.

\bibitem{RePEc:ucp:jpolec:doi:10.1086/664613}
E.~Budish.
\newblock The combinatorial assignment problem: Approximate competitive
  equilibrium from equal incomes.
\newblock {\em Journal of Political Economy}, 119(6):1061 -- 1103, 2011.

\bibitem{10.1145/3355902}
I.~Caragiannis, D.~Kurokawa, H.~Moulin, A.~D. Procaccia, N.~Shah, and J.~Wang.
\newblock The unreasonable fairness of maximum nash welfare.
\newblock {\em ACM Trans. Econ. Comput.}, 7(3), Sept. 2019.

\bibitem{DBLP:journals/corr/abs-2002-05119}
B.~R. Chaudhury, J.~Garg, and K.~Mehlhorn.
\newblock {EFX} exists for three agents.
\newblock {\em CoRR}, abs/2002.05119, 2020.

\bibitem{christodoulou2024fairtruthfulallocationsleveled}
G.~Christodoulou and V.~Christoforidis.
\newblock Fair and truthful allocations under leveled valuations, 2024.

\bibitem{christoforidis2024pursuitefxchoresnonexistence}
V.~Christoforidis and C.~Santorinaios.
\newblock On the pursuit of {EFX} for chores: Non-existence and approximations,
  2024.

\bibitem{feige2021tightnegativeexamplemms}
U.~Feige, A.~Sapir, and L.~Tauber.
\newblock A tight negative example for {MMS} fair allocations, 2021.

\bibitem{DBLP:journals/corr/abs-2109-08671}
Y.~Gafni, X.~Huang, R.~Lavi, and I.~Talgam{-}Cohen.
\newblock Unified fair allocation of goods and chores via copies.
\newblock {\em CoRR}, abs/2109.08671, 2021.

\bibitem{10.1016/j.artint.2021.103633}
M.~Ghodsi, M.~HajiAghayi, M.~Seddighin, S.~Seddighin, and H.~Yami.
\newblock Fair allocation of indivisible goods: Beyond additive valuations.
\newblock {\em Artif. Intell.}, 303(C), Feb. 2022.

\bibitem{10.1145/3140756}
D.~Kurokawa, A.~D. Procaccia, and J.~Wang.
\newblock Fair enough: Guaranteeing approximate maximin shares.
\newblock {\em J. ACM}, 65(2), Feb. 2018.

\bibitem{DBLP:journals/corr/abs-2008-08798}
R.~Mahara.
\newblock Existence of {EFX} for two additive valuations.
\newblock {\em CoRR}, abs/2008.08798, 2020.

\bibitem{10.1287/moor.2022.0044}
R.~Mahara.
\newblock Extension of additive valuations to general valuations on the
  existence of {EFX}.
\newblock {\em Math. Oper. Res.}, 49(2):1263–1277, Aug. 2023.

\bibitem{DBLP:journals/corr/PlautR17}
B.~Plaut and T.~Roughgarden.
\newblock Almost envy-freeness with general valuations.
\newblock {\em CoRR}, abs/1707.04769, 2017.

\bibitem{10.1016/j.artint.2023.104049}
M.~Seddighin and S.~Seddighin.
\newblock Improved maximin guarantees for subadditive and fractionally
  subadditive fair allocation problem.
\newblock {\em Artif. Intell.}, 327(C), Apr. 2024.

\bibitem{uziahu2023fairallocationindivisiblegoods}
G.~B. Uziahu and U.~Feige.
\newblock On fair allocation of indivisible goods to submodular agents, 2023.

\end{thebibliography}
\end{document}